%% file: main.tex
\documentclass{article}
\input{major/style.tex}

\title{Identifiability of AMP chain graph models}

\author{
Yuhao Wang \footnote{National University of Singapore. Email: \texttt{yohanna.wang0924@gmail.com}} \and

Arnab Bhattacharyya\footnote{National University of Singapore. Email: \texttt{arnabb@nus.edu.sg}}\and 
}
\date{}
\BeforeBeginEnvironment{appendices}{\clearpage}

\begin{document}

\maketitle

\begin{abstract}

We study identifiability of Andersson-Madigan-Perlman (AMP) chain graph models, which are a common generalization of linear structural equation models and Gaussian graphical models. AMP models are described by DAGs on chain components which themselves are undirected graphs. 

{For a known chain component decomposition}, we show that the DAG on the chain components is identifiable if the determinants of the residual covariance matrices of the chain components are monotone non-decreasing in topological order. This condition extends the equal variance identifiability criterion for Bayes nets, and it can be generalized from determinants to any super-additive function on positive semidefinite matrices. When the component decomposition is  unknown, we describe conditions that allow recovery of the full structure using a polynomial time algorithm based on submodular function minimization. We also conduct experiments comparing our algorithm's performance against existing baselines\footnote{Code is available at \hyperlink{https://github.com/YohannaWANG/DCOV}{https://github.com/YohannaWANG/DCOV}}.
\end{abstract}

\input{major/intro.tex}

\input{major/related.tex}

\input{major/prelim.tex}

\input{major/id.tex}

\input{major/expt.tex}

\input{major/conc.tex}

\bibliography{major/sample}


\end{document}

%% file: major/style.tex
\usepackage[english]{babel}
\usepackage[utf8x]{inputenc}
\usepackage[T1]{fontenc}

\usepackage[margin=1in]{geometry}
\usepackage{nicefrac} 
\usepackage{amssymb,amsmath,amsfonts,amsthm,dsfont}
\usepackage{mathtools, thmtools}
\usepackage{graphicx}
\usepackage{float}
\usepackage{enumitem}
\setlist[enumerate]{itemsep=0mm}

\newtheorem{theorem}{Theorem}[section]\numberwithin{equation}{section}
\newtheorem{corollary}[theorem]{Corollary}
\newtheorem{definition}[theorem]{Definition}

\newtheorem{fact}[theorem]{Fact}

\newtheorem{lemma}[theorem]{Lemma}

\usepackage{diagbox, eqparbox, hhline}
\usepackage{caption}

\usepackage{hyperref}

\hypersetup{
    colorlinks=true,
    linkcolor=blue,
    filecolor=magenta,      
    pdftitle={Overleaf Example},
    pdfpagemode=FullScreen,
    }

\usepackage{subfigure}
\usepackage{url}

\usepackage{thm-restate}

\usepackage{tikz}
\usepackage{graphicx}
\usepackage{xcolor}

\newcommand{\bbC}{\mathbb{C}}
\newcommand{\bbR}{\mathbb{R}}

\newcommand{\caC}{\mathcal{C}}
\newcommand{\caN}{\mathcal{N}}
\newcommand{\caD}{\mathcal{D}}

\DeclareMathOperator*{\Exp}{\mathbb{E}}  
\DeclareMathOperator*{\Cov}{Cov}  
\DeclareMathOperator{\Det}{Det}  
\DeclareMathOperator{\PA}{Pa}    
\DeclareMathOperator{\CH}{Ch}    
\DeclareMathOperator{\NE}{Ne}    
\DeclareMathOperator{\AN}{An}    
\DeclareMathOperator{\DE}{De}    
\usepackage{mathtools}
\newcommand{\defeq}{\vcentcolon=} 
\newcommand{\ignore}[1]{}

\usepackage{todonotes}

\newcounter{mynotes}
\DeclarePairedDelimiterX{\infdivx}[2]{(}{)}{%
  #1\;\delimsize\|\;#2%
}

\allowdisplaybreaks

\usepackage{nameref,cleveref}

\crefname{lemma}{lemma}{lemmas}
\Crefname{lemma}{Lemma}{Lemmas}
\crefname{thm}{theorem}{theorems}
\Crefname{thm}{Theorem}{Theorems}

\makeatletter
\newcommand{\myref}[1]{\cref{#1}\mynameref{#1}{\csname r@#1\endcsname}}
\newcommand{\Myref}[1]{\Cref{#1}\mynameref{#1}{\csname r@#1\endcsname}}

\usepackage[title]{appendix}
\usepackage{etoolbox}

\usepackage{xcolor}
\usepackage[linesnumbered,ruled,vlined]{algorithm2e}

\SetKwInput{KwInput}{Input}                
\SetKwInput{KwOutput}{Output}              

\usepackage{natbib} 
\usepackage{booktabs} 
\usepackage{tikz} 

\usepackage{tikz-cd}
\usetikzlibrary{shapes}
\usetikzlibrary{arrows}
\usetikzlibrary{er,positioning}

\usepackage{soul}
\setstcolor{red}

\newcommand*{\Scale}[2][4]{\scalebox{#1}{$#2$}}%

%% file: major/intro.tex

\section{Introduction}

Probabilistic graphical models offer architectures for modeling and representing uncertainties in decision making. 
From a computational standpoint, graphical representations {enable efficient algorithms} for inference, e.g., message passing, loopy belief propagation, and other variational inference methods \citep{KFL01}. They have found  applications in a wide range of domains, e.g., image processing, natural language processing and computational biology; see \citep{lauritzen1996graphical, KF09, WJ08} and references therein for examples.

{A typical application of} graphical models is to encode {causal} information. An influential article by \cite{Pearl95} elucidated how {\em Bayesian networks} can be used to represent causal processes and allow identification of causal effects. Bayesian networks are directed acyclic graphs (DAGs) in which the nodes represent variables of interest. Each node has a functional dependency on its parents, as determined by the graph. A popular way to substantiate Bayesian networks is as a {\em linear structural equation model (SEM)} where {variables that correspond to nodes in the graph} are a linear function of their parents' values  plus additive independent noise (often Gaussian) \citep{Bollen89, SGSH00}. \cite{HJMPS08} defined the more general {\em additive noise model} where each node is an arbitrary function of its parents with an additive independent noise. 

While Bayesian networks offer a clear conceptual way to model the causal structure of a system, they are in practice very hard to infer from data, as they require knowledge of how every single variable is generated. In applications involving hundreds of variables (e.g., in computational  biology), this requirement is unreasonable, particularly because at the end, we may only be interested in causal effects on a few target variables.
Furthermore, in {SEMs} modeled by Bayesian networks, the noise terms of different variables must be independent whereas in real-world systems, correlations can arise for various reasons (e.g., latent confounders).
\ignore{An example in medicine (\cite{sonntag2015chain}) when such a model might be appropriate is when we are modelling pain levels on different areas on the body of a patient. The pain levels can then be seen as correlated "geographically" over the body, and hence be modelled as a Markov network. Certain other factors do, however, exist that alters the pain levels locally
at some of these areas, such as the type of body part the area is located on or if local anaesthetic has been administered in that area and so on.}
An interesting middle ground is the notion of {\em chain graphs} \citep{LW98}. Here, the variable set is partitioned into {\em chain components}, and there is a DAG on these chain components. The variables inside each chain component, however, are connected by undirected edges, not directed ones. See Figure \ref{fig:chainintro} for an illustration. Thus, chain graph models interpolate between directed (causal) models and undirected (probabilistic) models.

\begin{figure*}
\centering    
\subfigure[$\mathcal{C}_1$]{\label{fig:cg_a}\includegraphics[width=33mm]{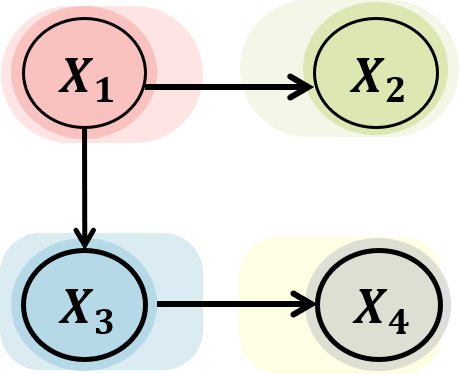}}
\hspace{12mm}
\subfigure[$\mathcal{C}_2$]{\label{fig:cg_b}\includegraphics[width=33mm]{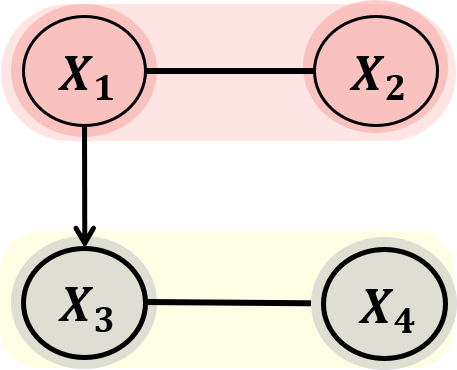}}
\hspace{12mm}
\subfigure[$\mathcal{C}_3$]{\label{fig:cg_c}\includegraphics[width=33mm]{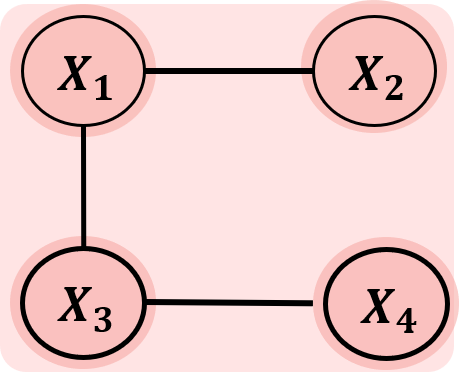}}
\caption{Chain graphs. Each shaded region is a maximal chain component.}
\vskip -0.5cm
\label{fig:chainintro}
\end{figure*}

There are several prevalent interpretations of chain graph models, namely the  Lauritzen-Wermuth-Frydenberg (LWF) \citep{LW98, frydenberg1990chain}, Alternative Markov Property or Andersson-Madigan-Perlman (AMP) \citep{andersson2001alternative}, and Multivariate Regression (MVR) \citep{cox1993linear}. They differ in the conditional independence relations implied by the graphical structure. In this work, we restrict ourselves to the AMP interpretation, which is the most natural one from a generative viewpoint. Let $\mathcal{C}$ be an AMP chain graph\footnote{See Section \ref{sec:prelims} for formal definitions.} on $n$ nodes. Suppose the nodes are partitioned into chain components $\{\tau\}$. Then, we say that {a random variable} $X \in \bbR^n$ is {\em generated by $\mathcal{C}$} if for every chain component $\tau$:
\begin{equation}\label{eqn:ampsem}
X_\tau = M_\tau X_{\PA(\tau)} + Z_\tau
\end{equation}
where $X_\tau$ is $X$ restricted to $\tau$, $\PA(\tau) = \{v: \exists u \in \tau, v\rightarrow_{\mathcal{C}} u\}$, $M_\tau$ is a matrix satisfying:
\[
(M_\tau)_{uv} \neq 0 \implies v \rightarrow_{\mathcal{C}} u,
\]
and $Z_\tau$ is an independent multivariate {Gaussian} drawn from $N(0, \Sigma_\tau)$ where $\Sigma_\tau$ satisfies:
\[
(\Sigma_\tau^{-1})_{uv} \neq 0 \implies u~\rule[.5ex]{1em}{0.4pt}_{\mathcal{C}}~ v
\]
The last condition ensures that $N(0, \Sigma_\tau)$ is Markovian with respect to the undirected induced subgraph $\mathcal{C}_\tau$ on $\tau$. One may also consider the additive noise AMP formulation where each 
\begin{equation}
X_\tau = f_\tau(X_{\PA(\tau)}) +Z_\tau,\label{eqn:addamp}
\end{equation} the noise $Z_\tau$ is as above, and the function $f_\tau$ is arbitrary, {provided it satisfies} the directed graph structure: 
\[\frac{(\partial f_\tau)_u}{\partial X_v} \neq 0 \implies v \rightarrow_{\mathcal{C}} u.\] 
On one hand, the directed edges of the AMP chain graph form a Bayesian network on the chain components. On the other hand, for each $\tau$, the undirected induced subgraph $\mathcal{C}_\tau$ describes a Gaussian graphical model for $X_\tau \mid X_{\PA(\tau)}$. 

In this work, we focus on the question of {\em identifiability} of chain graph models. That is, given knowledge of the distribution of $X$, can we recover the AMP chain graph $\mathcal{C}$ generating $X$? Moreover, can we recover $\mathcal{C}$ in polynomial time?  For Bayesian networks\footnote{For Gaussian graphical models, identifiability reduces to finding the inverse of the covariance matrix.}, the study of identifiability  has received sustained attention for more than two decades. By making faithfulness or related assumptions, many sets of researchers (e.g., \cite{SGSH00, chickering2002, ZS16, RU18}) have shown that the underlying DAG can be recovered {up to} its Markov equivalence class. This is quite unsatisfactory as the faithfulness assumption becomes too restrictive in the presence of finite sample error and the DAG is not uniquely identifiable. In a different line of work, \cite{PB14} showed that $\mathcal{C}$ is exactly identifiable for linear Gaussian {SEMs} if all the noise terms have equal variance. \cite{GH17, GH18} and \cite{PK20} established identifiability conditions for linear {SEMs} even with unknown heterogeneous error variances. Most recently, \cite{Park20} extended these conditions to additive noise models, while \cite{GDA20} further generalized to arbitrary Bayesian networks.  See also \cite{Eberhardt2017} and \cite{GZS19} for different perspectives on this line of work.

We extend these identifiability conditions from Bayesian networks to chain graphs. Our main contributions are:

\begin{enumerate}
\item[(i)]
\textbf{Additive noise AMP with known chain component decomposition}: 
We give a general class of identifiability conditions (generalizing the equal variance condition for linear {SEMs}) that imply identifiability of the DAG on a known collection of chain components. For instance, the DAG is identifiable if the determinant of the conditional covariance of a chain component $\tau$ given $\tau$'s parents is the same for all $\tau$. More generally, it is sufficient for this determinant to be monotonically non-decreasing with respect to a topological order on the chain components.  The same is true if the trace or the permanent satisfies the monotonicity condition.

\item[(ii)]
\textbf{AMP with unknown chain component decomposition}: We give an identifiability condition for recovering the chain components as well as the DAG for the standard AMP chain graph model. Informally, the requirement is quite natural: the variables in each chain component should be tightly correlated, while as a whole, each chain component should have large variance conditioned on its parents. More formally, the conditions are that:
\begin{itemize}
\item[(a)]
If $S$ is a proper subset of a chain component $\tau$:
\[
\det(\Cov(X_S \mid X_{S \setminus \tau}, X_{\PA(\tau)})) < 1
\]
\item[(b)]
$\det(\Cov(X_\tau \mid X_{\PA(\tau)}))$ is greater than 1 and monotonically non-decreasing with a topological order on the chain components $\tau$.

\end{itemize}
\end{enumerate}

In our conditions, the determinant of the covariance matrix of Gaussians plays a central role, and this is for good reason. If $X \sim N(0,\Sigma)$ is an $n$-dimensional Gaussian, then $\det(\Sigma)$ is the {\em generalized variance} of $X$ and is related to its {\em differential entropy}. Namely, the differential entropy of $X$ is $\frac12 (\log \det(\Sigma) + n \log(2\pi e))$;  see, e.g., \cite{KSG08, yunotes}. So, one can interpret condition (a) above as: If $S$ is a proper subset of $\tau$, its differential entropy conditioned on $\tau$ and $\tau$'s parents is smaller than a threshold. Similarly, the first part of condition (b) can be restated as: If $S$ equals $\tau$, the differential entropy of $S$ conditioned on its parents is larger than a threshold.

These identifiability conditions come with polynomial time algorithms. Notably, our algorithm for recovering the chain components in (ii) above involves a non-trivial submodular function minimization, in contrast to the more straightforward algorithms known for identifying linear SEMs and Bayesian networks \citep{Park20, GDA20} under analogous conditions.
\ignore{
As a subclass of probabilistic graphical model, Bayesian networks do however have a major limitations of representing the asymmetric causal relationships. Chain graphs (CGs) are a natural generalizations of both causation and correlation and is tried to solve this shortcomings. Depending on the interpretation of the undirected connections, there are several prevalent CGs, namely, Lauritzen-Wermuth-Frydenberg (LWF) \cite{10.2307/4616142}, Alternative Markov Property (AMP)\cite{andersson2001alternative} , and Multivariate Regression (MVR) \cite{cox1993linear}, distinguished by the separation criteria. In detail, MVR interpretation  treat the undirected edge with a strong intuitive meaning. $X \leftrightarrow Y$ can be replaced by a Bayesian network of $X \leftarrow H \rightarrow Y$, with $H$ the hidden common cause. LWF and AMP interpretation shares the similar combination mechanism of Bayesian networks and Markov networks but different separation criteria. In this work, we are interested in the AMP interpretation and addresses the following question: under what assumptions on the data generation process can one infer the AMP-CG from the joint distribution? We discuss how the equal determinant assumption contribute to the AMP-CGs identifiability. We further provide a practical algorithm that recovers the AMP-CGs from finite sample data. Experiments on both simulated and real data support the theoretical findings. 
}
\subsection{Technical Overview}\label{sec:tover}
In this section, we describe some of the intuition behind our identifiability conditions.

\begin{figure}
\centering    
\subfigure[$\mathcal{C}_1$]{\label{fig:chain_a}\includegraphics[width=30mm]{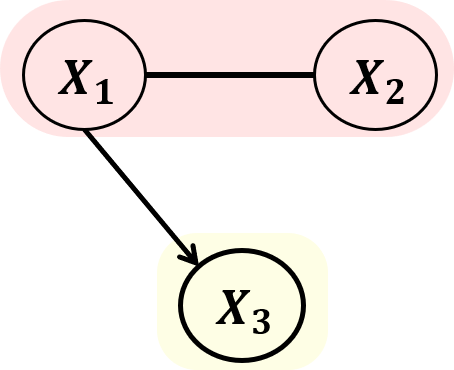}}
\hspace{20mm}
\subfigure[$\mathcal{C}_2$]{\label{fig:chain_b}\includegraphics[width=30mm]{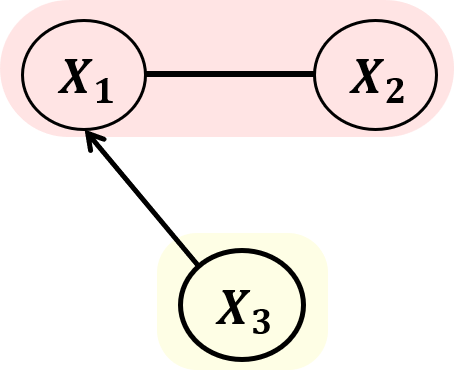}}
\caption{Chain graph identifiability}
\vskip -0.5cm
\label{fg:1}
\end{figure}
\paragraph{Known chain components.} Consider Figure \ref{fg:1} which shows two chain graphs $\mathcal{C}_1$ and $\mathcal{C}_2$; the question is to determine which of these graphs is generating a given joint distribution $(X_1, X_2, X_3)$. In $\mathcal{C}_1$, let $\begin{pmatrix} X_1\\  X_2 \end{pmatrix} \sim \mathcal{N}(0, \Sigma_1)$, and $X_3 = \beta_{1} X_1 + Z$, where $\beta_1 \neq 0$ and $Z \sim \mathcal{N}(0, \sigma^2)$. 
In $\mathcal{C}_2$, let $\begin{pmatrix} X_1\\  X_2 \end{pmatrix} = \begin{pmatrix} \beta_{2} \\0 \end{pmatrix}X_3 + Z$ where $\beta_2 \neq 0$, $Z \sim \mathcal{N}(0, \Sigma_2)$ and $X_3 \sim \caN(0,\sigma^2)$. Assume $\Det(\Sigma_1) =\Det(\Sigma_2)  = \sigma^2$, so that in both models, the determinant of the covariance of each chain component conditioned on its parents is $\sigma^2$.

We claim that in this case, one can distinguish between $\mathcal{C}_1$ and $\mathcal{C}_2$ based on the generated distribution. Our algorithm first finds the chain component $\tau$ minimizing $\det(\Cov(X_\tau))$. Note that for $\mathcal{C}_1$, using the independence of $Z$:
{
\[\Scale[0.9]
{\Cov (X_3)  = \Cov(\beta_{1}X_1 + Z) = \beta_1^2 \Cov(X_1) + \Cov(Z) \succ \Cov(Z),}    
\]}
assuming\footnote{In this work, we make the assumption everywhere that all covariance matrices are strictly positive definite.} that $\Cov(X_1) \succ 0$. Hence, 
$\det(\Cov(X_3)) > \det(Z) = \sigma^2 = \det(\Cov(X_{12}))$.
On the other hand for 
$\mathcal{C}_2$, $\det(\Cov(X_{12})) > \sigma^2 = \det(\Cov(X_3))$.
Thus, the chain component with the smallest determinant of the covariance can be identified as the first in a topological ordering. 
This can be understood as the uncertainty level of the parents is less than its children. Once the first chain component is known, we can select the second by choosing the one that minimizes the determinant of its covariance conditioned on the first chain component, and so on. It suffices to find the topological order because as described in Appendix A of \cite{GDA20}, one can identify the directed edges by standard variable selection methods.

Note that the only property we used of the determinant is that $\det(A+B) > \det(A)$ if $B$ is strictly positive definite. This property holds not only for the determinant but for many natural matrix functions. For example for any $i$, the diagonal entries $(A+B)_{ii} > A_{ii}$  when $A$ and $B$ are positive definite. Carrying out the same logic as above but now using projection to diagonal entries instead of determinants implies that the chain component DAG is identifiable when all the individual variables have equal variance, extending the result of \cite{PB14} to chain graphs. In fact, there is a large class of functions called ``generalized matrix functions'' that satisfy the desired super-additivity condition and hence result in identifiability conditions for the DAG on chain components.

\paragraph{Unknown chain components.}
Consider again $\mathcal{C}_1$ from \Cref{fg:1}, but suppose now that we do not have the chain component partitioning. Let $(X_1, X_2, X_3)$ be generated as described above. In addition to imposing the condition that $\det(\Sigma_1) = \sigma^2$, we now also require that: (i) $\det(\Cov(X_1|X_2))$ and $\det(\Cov(X_2 \mid X_1))$ are\footnote{$\det(\Cov(X_2 \mid X_1))$ is well defined, since $(X_1, X_2)$ are jointly {Gaussian}, and hence, for any choice of $x_1$, $\Cov(X_2 \mid X_1 = x_1)$ is the same.} strictly less than $1$, and (ii) $\sigma^2$ is strictly greater than $1$. 

Now, we can show that \[\det(\Cov(X_{12})) = \min_{S \subseteq \{1,2,3\}} \det(\Cov(X_S)).\] Observe that $\det(\Cov(X_3)) > \det(\Cov(X_{12})) $ already follows from the earlier discussion. We now compare $\det(\Cov(X_{12}))$ to $\det(\Cov(X_1))$ and $\det(\Cov(X_2))$. We use the fact that:
\[\det(\Cov(X_{12})) = \det(\Cov(X_1)) \cdot \det(\Cov(X_2 \mid X_1)).\]
This follows from standard facts about multivariate {Gaussians}. From our assumption $\det(\Cov(X_2 \mid X_1))<1$, we get that $\det(\Cov(X_1)) > \det(\Cov(X_{12}))$. The same holds for $\det(\Cov(X_2))$. 
Finally, we need to show that $\det(\Cov(X_{123})) > \det(\Cov(X_{12}))$. Again, we can invoke the above fact:
\[
\det(\Cov(X_{123})) = \det(\Cov(X_{12})) \cdot \det(\Cov(X_3 \mid X_{12})).
\]
Our conclusion follows from the assumption $\sigma^2 > 1$. 

For a general chain graph, it similarly follows that the non-empty set $S$ minimizing $\det(\Cov(X_S))$ is the topologically smallest. We can identify the next component by conditioning on the components already discovered, which results in a {Gaussian} on the rest, and then finding a non-empty subset with conditional covariance matrix of smallest determinant. This algorithm can be implemented efficiently. The reason is that for any positive definite $n\times n$-matrix $M$, the function $F(S) = \log\det (M[S, S])$, where $M[S, S]$ is the submatrix on rows and columns indexed by $S \subseteq [n]$, is {\em submodular}. $F$, as noted earlier, corresponds to the differential entropy of a Gaussian vector with covariance $M$, which is a submodular function, plus an additional modular term. The problem of submodular function minimization has a long and rich history, beginning with the seminal works of \cite{GLS81, GLS12} and continuing to the current day \citep{IFF01, Sch00, LSW15, DVZ18, Jiang21}.  Thus, we can invoke any of these known polynomial-time algorithms for submodular function minimization to recover the chain components in topological order.

\ignore{
It is useful to interpret the situation geometrically. Recall that if a positive definite $M$ has a Cholesky decomposition $L L^T$, and the rows of $L$ are given $a_1, \dots, a_k$, then $\det(M)$ is the square of the volume of the parallelepiped spanned by $a_1, \dots, a_k$. Now, suppose $\Cov(X_{12}) = \begin{pmatrix} \sigma^2 & (1-\epsilon)\sigma^2 \\(1-\epsilon)\sigma^2 & \sigma^2 \end{pmatrix}$ while $X_3 =  X_1 + Z$ where $Z \sim \caN(0, \sigma_3^2)$ is independent of $X_1$. Then, it can be seen that the three Cholesky vectors for $\Cov(X_{123})$ are as in \Cref{fg:2}. 
If $\sigma_3 > 1$, the volume spanned by the three vectors altogether is larger than the area spanned by the first two (shaded region); this corresponds to requirement (i) above. Also, the first two vectors of length $\sigma$ form a very small angle with each other, so that the shaded region's area is smaller than $\sigma$; this corresponds to requirement (ii) above. Finally, the length of the third vector is $\sqrt{\sigma^2 + \sigma_3^2} > \sigma$ if $\sigma_3 \geq \sigma$; this corresponds to the requirement that we imposed when the chain graphs are known.

\begin{figure}[t]
\centering
\includegraphics[width=0.35\textwidth]{fig/Fig1.png}
\caption{Stuff}
\vskip -0.3cm
\label{fg:2}
\end{figure}
}

%% file: major/related.tex

\subsection{Related Work}
Chain graph models contain both directed and undirected edges and can be used to represent both association and causation in real-world applications \citep{sonntag2016chain}. 
The three following interpretations are the best known in the literature:
LWF \citep{lauritzen1989graphical, wermuth1990substantive, frydenberg1990chain} which generalizes both Markov random fields and Bayesian networks; 
AMP \citep{andersson2001alternative, andersson2006characterizing} which directly extends the DAG Markov property; and  MVR \citep{cox1993linear, cox2014multivariate} which originates from viewing undirected edges as representing hidden common causes.

The literature on learning pure DAG models is vast. 
One popular approach is to exploit the constraints imposed by Markov structure, e.g., the PC algorithm and its variants, like Fast Causal Inference (FCI), Really Fast Causal Inference (RFCI) and Cyclic Causal Discovery (CCD) \citep{spirtes2000, SGSH00, Rich13,  colombo2011learning, CMH13, HD13, CM14}. Another important class of algorithms aims to maximize a score function over the space of DAG's, such as Greedy Equivalence Search (GES) \citep{chickering2002, RGSG17, nandy2018high} and a recent line of work that formulates score maximization as a continuous optimization problem (e.g., \citep{zheng2018dags, zheng2020learning, WGY20}). This latest direction has resulted in algorithms that  learn the DAG structure with deep learning methods (e.g.,  \cite{yu2019dag, LBDL20, wang2020causal}).

A probability distribution may be Markov with respect to many Bayes networks; so for exact identifiability, one needs to impose more structural constraints on the DAG model. For Structural Equation Models (SEM's), identifiability can be established by leveraging asymmetries between variable pairs  \cite{shimizu2006,mooij2016},  restricting SEMs to having additive noise, such as linear non-Gaussian acyclic model (LiNGAM) \citep{shimizu2006}, general additive noise models \citep{peters2014causal}, Post-nonlinear model (PNL) \citep{zhang2016estimation}, or equal and unknown error variance \citep{PB14, GH17,Eberhardt2017, GH18, chen2019causal, GZS19, PK20, Park20, GDA20}.
 
 AMP chain graphs, our focus in this work, have been less widely studied than pure DAG models and more in the statistics literature than computer science. Informally speaking, \cite{pena2015every} showed that any AMP model can be viewed as arising from a DAG causal model subject to selection bias. 
\cite{levitz2001separation} introduced a pathwise separation criterion to characterize conditional independence relations  in AMP chain graphs. \cite{roverato2005unified, studeny2009two, pena2017identification} studied the equivalence classes of chain graph models, and \citep{pena2018reasoning} provided a factorization for positive distributions that are Markov with respect to an AMP chain graph. \cite{drton2009discrete} showed that the AMP conditional independence relations may lead to non-smooth models for discrete variables.
\cite{pena2014marginal, pena2016alternative} investigated extensions to the AMP model, e.g., the marginal AMP model (MAMP) that is a common generalization of  AMP and MVR. When the chain graph structure is known, \cite{DE06} proposed an algorithm for maximum likelihood estimation of the model parameters.
\cite{pena2012learning, pena2014learning, pena2016learning} proposed \textsc{PC-like}, a constraint based algorithm under faithfulness assumptions for learning the structure of AMP and MAMP models. Pe{\~na} also designed a score-based algorithm for AMP model structure learning similar to the work on additive noise models \citep{pena2017learning} and an algorithm based on answer set programming \citep{pena2016alternative}. 
Recently, \cite{javidian2020amp} solved the problem of efficiently finding minimal separating sets in AMP chain graphs and obtained a new decomposition-based structure learning algorithm called \textsc{Lcd-AMP}.  


%% file: major/prelim.tex

\section{Notations and Preliminaries}\label{sec:prelims}
\paragraph{Probability.} We need the following useful fact about conditional covariance. The proof is a
simple generalization of the standard proof for law of total variance.
\begin{fact}[Law of Conditional Covariance]\label{fact:tcov} {If X, Y, Z are random variables with strictly positive distributions with each component having finite second moment, then:}
\[
\Cov_X(X \mid Y) = \Exp_Z[\Cov_X(X \mid Y,Z) \mid Y] + \Cov_Z(\Exp_X[X \mid Y, Z] \mid Y).
\]
\end{fact}

The following result yields a very useful decomposition for covariance of normal distributions. 
\begin{fact}\label{FACT:FACTORIZATION}
If $X=(X_A, X_B)$ is distributed jointly as a {Gaussian} $\mathcal{N}(0, \Sigma)$, then:
\[
\det(\Cov(X)) = \det(\Cov(X_A)) \cdot \det(\Cov(X_B \mid X_A))
\]
where $\Cov(X_B \mid X_A) = \Cov(X_B \mid X_A = x_A)$ is independent of $x_A$.
\end{fact}

\begin{proof}
It is well-known that if $\Cov(X) = \begin{pmatrix} \Sigma_{AA} & \Sigma_{AB}\\ \Sigma_{BA} & \Sigma_{BB}\end{pmatrix}$:
\[
\Cov(X_A) = \Sigma_{AA} \qquad \Cov(X_B \mid X_A) = \Sigma_{BB} - \Sigma_{BA}\Sigma_{AA}^{-1} \Sigma_{AB}
\]
On the other hand, it follows from the properties of Schur complement that: $\det(\Sigma) = \det(\Sigma_{AA}) \cdot \det(\Sigma_{BB} - \Sigma_{BA}\Sigma_{AA}^{-1}\Sigma_{AB})$. The result follows.
\end{proof}

\paragraph{Chain Graphs.}
Following conventions in the field, a variable is denoted by an uppercase letter, e.g., $X$, and its value is denoted by the corresponding lowercase letter, $z \in Z$, where $Z$ is the
state space of $X$. Graphs in this paper contain both directed ('$\rightarrow$') and undirected ('---') edges. Below we will further invoke the most {central definitions and notations} used in this paper. For a general account, we refer the reader to \cite{lauritzen1996graphical} and \cite{edwards2012introduction}.

{A chain graph $\mathcal{C}$ 
consists of a \emph{vertex set} $V$ and an \emph{edge set} $E \in V \times V$.
A \emph{path} in $\mathcal{C}$ is a sequence of distinct vertices $<v_0,\dots, v_n>$ such that $v_{i-1}$ and $v_i$ are adjacent for all $1 \leq i \leq k$, and is called a \emph{cycle} if $v_n = v_0$. Moreover, a \emph{semi-directed cycle} exists if $v_1 \rightarrow v_2$ is in $\mathcal{C}$ and $v_i \rightarrow v_{i+1}$, $v_i \longleftrightarrow v_{i+1}$ or $v_i - v_{i+1}$ is in $\mathcal{C}$ for all $1 < i <n$.
Chain graph $\mathcal{C}$ is a graph has no semi-directed cycles.}
Two vertices joined by an edge are called \emph{adjacent}. For vertices $(u, v) \in E$ but $(v, u) \not\in E$, we write $u \rightarrow v $, where vertex $u$ is a \emph{parent} of $v$. If both $(u, v) \in E$ and $(v,u) \in E$, we denote it by $u - v$, which means $u$ is a \emph{neighbor} of $v$. The vertex set of a chain graph can be partitioned into \emph{chain components} $\tau \mid \tau \in \mathcal{T}, (V = \cup_{(\tau \in \mathcal{T})} \tau)$. Edges within chain components are undirected whereas edges between two chain components are directed. 
For any subset $S$, the set of parents of $v$ is denoted as $\PA(v) \defeq \{ v \in V \setminus S \mid v \rightarrow s \in \mathcal{C} \text{ for some } s \in S\}$, the set of $children$ of $v$ is denoted as $\CH(v) \defeq \{ v \in V \setminus S \mid s \rightarrow v \in \mathcal{C} \text{ for some } s \in S\}$, the set of $neighbours$ is denoted as $\NE(v) \defeq \{ v \in V \setminus S \mid v - s \in \mathcal{C} \text{ for some } s \in S\}$. A chain graph with no directed edges is an undirected graph (UG), while a chain graph with no undirected edges is a DAG. 
If there exists a directed path between chain components $k \rightarrow \dots \rightarrow j$,
then $k$ is an $ancestor$ of its $descendant\:\: j$. The set of $ancestors$ and $descendants$ are denoted as $\AN (j)$ and $\DE (k)$. 
A $source$ node is any node $X_{\tau}$ such that $\PA (X_{\tau}) = \emptyset$. 
A $sink$ node is any node $X_{\tau}$ such that $\CH (X_{\tau}) = \emptyset$. 
Ancestral set is any set $A \in V$ such that $X_{\tau} \in A \Longrightarrow \PA(\tau) \in A$.  
The chain components $\tau$ of a chain graph are the connected components of the undirected graph obtained by removing all directed edges from the chain graph. In a DAG, all chain components are singletons.
For $S \subseteq V$, $\mathcal{C}_S$ denotes the induced subgraph on $S$.

By taking into account the directed connections of chain components, AMP-chain graphs admits a topological ordering of its chain components. For statistical identifiability of chain graph $\mathcal{C}$, we will consider it sufficient to learn the partition into chain components $\tau_1, \dots, \tau_t$, and a topological ordering $ \prec$ such that 
$\tau_j \rightarrow \tau_k \Longrightarrow \tau_j \prec \tau_k$. One can learn the directed and undirected edges using standard parameter estimation algorithms.

\paragraph{Matrix Algebra.}
Our identifiability condition in the case of known chain components is in terms of positive and super-additive families, which we define next.
\begin{definition}
Let $\bbC_n$ denote the cone of $n\times n$ positive semidefinite matrices. We say that
a real-valued function $d_n: \bbC_n \to \bbR$ is {\em positive and super-additive} if: (i) $d_n(A) >0$ for all positive definite matrices $A$, and (ii)  for all positive semidefinite matrices $A, B$:
\[d_n(A+B) \geq d_n(A) + d_n(B).\]
A \emph{positive and super-additive family} is a collection of functions $f_n: \bbC_n \to \bbR$,
each of which is positive and super-additive.
\end{definition}
We have several examples of families of positive and super-additive functions:
\begin{itemize}
\item
Clearly, the projection on any diagonal element and the matrix trace function are positive and super-additive.
\item
By Minkowski's determinant theorem (see, e.g., \cite{marcus1992survey}), it known that for all $A, B \in \bbC_n$:
$(\det(A+B))^{1/n} \geq (\det(A))^{1/n} + (\det(B))^{1/n}$. Hence, $\{\det^{1/n}: \bbC_n \to \bbR\}$ is positive and super-additive.
\item
For $\chi$ an irreducible character on a subgroup $H$ of $S_n$ (the permutation group on $n$ elements),  define the {\em generalized matrix function} with respect to $H$ and $\chi$ as:
\[d_\chi^H(A) = \sum_{\sigma \in H} \chi(\sigma) \prod_{i=1}^n a_{i,\sigma(i)}\]
where $A = (a_{i,j})$. \cite{schur1918endliche} showed that $d_\chi^H(A) > 0$ for all positive definite $A$. It is also known (e.g., \cite{merris1997multilinear}, p.~228) that they satisfy the super-additivity condition. 
Hence, the determinant\footnote{The super-additivity of the determinant is also directly implied by the super-additivity of $\det^{1/n}$.}, permanent, and the Hadamard matrix function (product of diagonal entries) all form positive and super-additive families.
\end{itemize}

\ignore{

 and is equivalent to the set of conditional independences satisfied by the following recursive normal linear simultaneous equations model:
\begin{equation}
    X_{\tau} = \beta_{\tau}X_{pa_{D}(\tau)} + N _\tau,   \tau \in \mathcal{T}
\label{eq:1}
\end{equation}
With $X \equiv (X_v | v \in V) 	\equiv (X_\tau |\tau \in \mathcal{T})$ is a normal random vector in $\mathcal{R}^{v}$, $\beta_{\tau}$ is a $\tau \times pa_{D}(\tau)$ matrix satisfying $v \not\in pa_{G}(u) \Rightarrow (\beta_{\tau})_{uv} = 0$. 
Assume the noise term follows the distribution of $N \sim (0,  \Sigma_{\tau}), \Sigma_{\tau} \in \boldsymbol{P}(G_{\tau})$, 
where $\boldsymbol{P}(G_{\tau})$ is a set of all positive semidifinite $\tau \times \tau$ covariance matrices such that $\mathcal{N}(0, \Sigma_{\tau})$ is a global $G_{\tau}$ Markovian.
In matrix form, equation \eqref{eq:1} is given by $X = BX + N$, where $B = [B_1| \dots |B_d]$ is a weighted adjacency matrix and $N = (N_1| \dots |N_d)$ is a noise vector with independent elements. Equation \eqref{eq:1} admits the unique solution $X  = (I-B)^{-1}N$.
And the observed random vector $X = (X_1, \dots, X_n)$ is centered ($\mathbb{E}[X] = \mathbb{E}[(I-B)^{-1}N] = 0$) with the covariance $  \Sigma = \mathbb{E}[XX^T] = (I-B)^{-1}\Omega ((I-B)^{-1})^T $, where $\Omega = 1/m\sum_{i=1}^{m}n_in_i^T$ is the sample covariance matrix. It is well-known that if $\Omega$ is invertible, we have the precision matrix $ \Omega^{-1} = \mathbb{E}[NN^T]^{-1}$.}

%% file: major/id.tex
\section{Identifiability with known chain component decomposition}\label{sec:id}
In this section, we give a general class of conditions which are sufficient to ensure that the DAG structure 
of the chain graph is identifiable from data generated by it. Here, the chain component decomposition $\caD$
is already known to the algorithm. $\caD$ consists of $t$ disjoint maximal chain components that partition the 
variable set.

We formulate our results for general AMP chain graph models. They will immediately imply the conditions for additive noise AMP models mentioned in the Introduction.

\begin{algorithm}
$\mathcal{A}, P \leftarrow \emptyset$\;
$i \leftarrow 0$\;
\While{$|\mathcal{A}| \neq t$}{
$\tau_i \leftarrow \arg\min_{\tau \in \mathcal{C} \setminus \mathcal{A}} d_{|\tau|}(\Exp[\Cov(X_{\tau} \mid X_P)])$\;
$\mathcal{A} \leftarrow \mathcal{A} \cup \{\tau_i\}$\;
$P \leftarrow P \cup \tau_i$\;
$i \leftarrow i+1$\;
}
Return the ordering $(\tau_1, \dots, \tau_t)$
\caption{\textcolor{black}{Our algorithm} for learning the topological order of a chain graph with chain component decomposition $\mathcal{D}$ of size $t$.}
\label{algo:pop1}
\end{algorithm}
\vspace{-2mm}

\ignore{
Let $X_{\tau}$ be a finite family of random variables. We address the following problem:
\begin{problem}
    Given i.i.d. samples from a joint distribution $P^{(X_i), i\in V}$, infer the true DAGs of the process that generated the data.
\end{problem}

\noindent As a special class of chain graph, DAG structures can be uniquely identified when it follows a structural equation model with equal error variance. Detailed theory is shown below:

\begin{theorem}[\cite{chen2019causal} \cite{peters2014identifiability} ]
Under a structural equation model with directed acyclic graph $\mathcal{G}_0$ of the form $X_j = \sum_{k \in \textbf{PA}_{j}^{\mathcal{G}_0}}\beta_{jk}X_k + N_j, (j=1,\dots,n)$. Let $\mathcal{L}(X)$ represent $\mathcal{G}(B_X)$ with $X \sim (B_x, \sigma_x^2)$ and $\mathcal{L}(Y)$ represent $\mathcal{G}(B_Y)$ with $Y \sim (B_y, \sigma_Y^2)$. If $\sigma_x^2 = \sigma_Y^2$, then  $\mathcal{G}(B_X)$ and $\mathcal{G}(B_Y)$ lead to the same distribution.
\label{th:1}
\end{theorem}

\noindent The inductive proof of Theorem \eqref{th:1} is available in \cite{chen2019causal} \cite{peters2014identifiability}. Work in \cite{peters2014identifiability} prove this theorem by leveraging the Markov condition and causal minimality, while \cite{chen2019causal} implies this identifiability condition by an ordering among conditional variances. 
In short, by recover the topological ordering among variables, the DAG structures can be thus identified.
However, the above identifiable condition is too restrictive. \cite{park2020identifiability} provides a new identifiability condition for ANMs with any forms of dependency functions and heterogeneous error distributions. A bi-variate example in distinguishing between $X_1 \rightarrow X_2$ and $X_2 \rightarrow X_1$ is given in Appendix \ref{appendix:A}. The general theorem is given below:
\begin{theorem}[\cite{park2020identifiability}](\textbf{ANMs}) 
Let $P(X)$ be generated from an ANM with DAG G and true ordering $\pi$. Suppose that causal minimality holds. DAG structure is uniquely identifiable if either of the two following conditions is satisfied: For any node $j = \pi_m \in V, k \in De(j)$, and $\ell \in An(j) $
\begin{equation}
\begin{split}
    \text{Forward stepwise selection} & : \sigma_j^2 < \sigma_k^2 + \mathbb{E}[\operatorname{Var}(\mathbb{E}[X_k | X_{Pa(k)}] \:|\: X_{\pi_1}, \dots, X_{\pi_{m-1}}  )], or \\
    \text{Backward stepwise selection} &: \sigma_j^2 > \sigma_{\ell}^2 - \mathbb{E}[\operatorname{Var}( \mathbb{E}[X_{\ell}| X_{\pi_1}, \dots, X_{\pi_m} 	\setminus X_{\ell}]\: | \:X_{Pa(\ell)})]
\end{split}
\label{th:2}
\end{equation}
\end{theorem}

The above theorem generalized the equal error variance assumption in \cite{peters2014identifiability} for DAG structure learning. Now we extend the above two theories into the chain graph structure learning condition.

\begin{problem}
    Given i.i.d. samples from a joint distribution $P^{(X_i), i\in V}$, infer the true AMP-CGs of the process that generated the data.
\end{problem}

Similarly, because AMP-CG shares the similar data generation mechanism of DAG \cite{andersson2001alternative}. We will prove our intuition that chain graph structures are identifiable by recovering the topological ordering among the determinant of chain components.
Specifically, we extend the above theorem into the chain graph condition and prove that chain graph structures can be identified within two steps. In the first step, the algorithm discover the ordering component wisely using the equal determinant assumption. In the second step, the algorithm estimate the directed edges among chain components. In the end, we extend the equal determinant assumption into a generalized form. Our theorem lists below:}
\ignore{
\begin{theorem}
Given two chain graphs $\mathcal{G}_X$ and $\mathcal{G}_Y$, let $X \sim (D_X, \Sigma_X)$ and $Y \sim (D_Y, \Sigma_Y)$ with both $G_X$ and $G_Y$ are AMP-CGs. If $\Det (\mathbb{E}(\Cov(X) )) = \Det (\mathbb{E}(\Cov (Y)))$, then $\mathcal{G}(X) = \mathcal{G}(Y)$, $D_X = D_Y$, and $\Sigma_X = \Sigma_Y$.
\label{th:3}
\end{theorem}
}

\begin{theorem}\label{THM:ID}
Suppose the random variable $X$ is generated by an AMP-CG $\caC$ with known chain component decomposition $\caD$. Then, $\caC$ is identifiable from $P$ if there {exists} a  topological ordering $\pi$ of $\caC$ and a positive and super-additive family $\{d_n: \bbC_n \to \bbR\}$ such that:
\begin{small}
\begin{align}\label{eq:thmcond}
d_{|\tau|}\left(\Exp_{X_{\PA(\tau)}}\Cov_{X_\tau} (X_{\tau} \mid X_{\PA(\tau)}) \right)  \leq  d_{|\tau'|}\left(\Exp_{X_{\PA(\tau')}}\Cov_{X_{\tau'}}(X_{\tau'} \mid X_{\PA(\tau')})\right)
\end{align}
\end{small}
for any two chain components $\tau, \tau'$ where $\tau \prec_\pi \tau'$.
\label{th:4}
\end{theorem}
\vspace{-1mm}
\begin{proof}
We show that Algorithm \ref{algo:pop1} recovers the chain graph under the assumptions of the Theorem \ref{th:4}. This follows immediately from the following lemma, as it shows that at every step $i$, the algorithm chooses as $\tau_i$ a chain component whose parents are contained in the current $\mathcal{A}$. 

\begin{lemma}
Let $\mathcal{A}$ be an ancestral set of chain components, and let $P = \{v: v \in \tau \in \mathcal{A}\}$. Assume the condition (\ref{eq:thmcond}) above. Suppose $\tau_1$ and $\tau_2$ are chain components in $\caD \setminus \mathcal{A}$ such that $\PA(\tau_1) \subseteq \mathcal{A}$ but $\PA(\tau_2) \not \subseteq \mathcal{A}$. Then:
\[
d_{|\tau_1|}\left(\Exp_{X_P} \Cov_{X_{\tau_1}}(X_{\tau_1} \mid X_P)\right) < d_{|\tau_2|}\left(\Exp_{X_P} \Cov_{X_{\tau_2}}(X_{\tau_2} \mid X_P)\right)
\]
\end{lemma}
\begin{proof}
Note that $\tau_1$ must precede $\tau_2$ in the topological ordering $\pi$, and hence (\ref{eq:thmcond}) holds with $\tau = \tau_1$ and $\tau' = \tau_2$.
We invoke the law of conditional covariance (Fact \ref{fact:tcov}).
\begin{small}
\begin{align*}
&\Exp_{X_P} \Cov_{X_{\tau'}}(X_{\tau'} \mid X_P)\\
= & \Exp_{X_P} \Exp_{X_{\PA(\tau')}}\left[\Cov_{X_{\tau'}}(X_{\tau'} \mid X_P, X_{\PA(\tau')})\mid X_P\right] + \Exp_{X_P}\Cov_{X_{\PA(\tau')}}\left(\Exp_{X_{\tau'}}(X_{\tau'} \mid X_P, X_{\PA(\tau')})\mid X_P\right) \\
=& \Exp_{X_P} \Exp_{X_{\PA(\tau')}}\left[\Cov_{X_{\tau'}}(X_{\tau'} \mid X_{\PA(\tau')})\mid X_P\right] + \Exp_{X_P}\Cov_{X_{\PA(\tau')}}\left(\Exp_{X_{\tau'}}(X_{\tau'} \mid X_{\PA(\tau')})\mid X_P\right) \\
=& \Exp_{X_{\PA(\tau')}}\left[\Cov_{X_{\tau'}}(X_{\tau'} \mid X_{\PA(\tau')})\right] + \Exp_{X_P}\Cov_{X_{\PA(\tau')}}\left(\Exp_{X_{\tau'}}(X_{\tau'} \mid X_{\PA(\tau')})\mid X_P\right)
\end{align*}
\end{small}
The second equality follows from the fact that $X_{\tau'}$ is independent of $X_P$, conditioned on $X_{\PA(\tau')}$. Now, note that the second term in the last line above is positive definite if $P$ does not contain $\PA(\tau')$.  Therefore, using the fact that $d_{|\tau'|}$ is positive and super-additive:
\ignore{
\begin{align*}
&d_{|\tau'|}\left(\Exp_{X_P} \Cov_{X_{\tau'}}(X_{\tau'} \mid X_P)\right)\\
&\geq d_{|\tau'|}\left(\Exp_{X_{\PA(\tau')}}\left[\Cov_{X_{\tau'}}(X_{\tau'} \mid X_{\PA(\tau')})\right]\right) +\\
&\qquad \quad d_{|\tau'|}\left(\Exp_{X_P}\Cov_{X_{\PA(\tau')}}\left(\Exp_{X_{\tau'}}(X_{\tau'} \mid X_{\PA(\tau')})\mid X_P\right)\right)\\
&> d_{|\tau'|}\left(\Exp_{X_{\PA(\tau')}}\left[\Cov_{X_{\tau'}}(X_{\tau'} \mid X_{\PA(\tau')})\right]\right)\\
&\geq d_{|\tau|}\left(\Exp_{X_{\PA(\tau)}}\left[\Cov_{X_{\tau}}(X_{\tau} \mid X_{\PA(\tau)})\right]\right)\\
&= d_{|\tau|}\left(\Exp_{X_{P}}\left[\Cov_{X_{\tau}}(X_{\tau} \mid X_{P})\right]\right)
\end{align*}
}
\begin{small}
\begin{alignat*}{2}
&d_{|\tau'|}\left(\Exp_{X_P} \Cov_{X_{\tau'}}(X_{\tau'} \mid X_P)\right) \\
&\geq d_{|\tau'|}\left(\Exp_{X_{\PA(\tau')}}\left[\Cov_{X_{\tau'}}(X_{\tau'}  \mid X_{\PA(\tau')})\right]\right) +
d_{|\tau'|}\left(\Exp_{X_P}\Cov_{X_{\PA(\tau')}}\left(\Exp_{X_{\tau'}}(X_{\tau'} \mid X_{\PA(\tau')})\mid X_P\right)\right)\\
&> d_{|\tau'|}\left(\Exp_{X_{\PA(\tau')}}\left[\Cov_{X_{\tau'}}(X_{\tau'} \mid X_{\PA(\tau')})\right]\right)\\
&\geq d_{|\tau|}\left(\Exp_{X_{\PA(\tau)}}\left[\Cov_{X_{\tau}}(X_{\tau} \mid X_{\PA(\tau)})\right]\right)
= d_{|\tau|}\left(\Exp_{X_{P}}\left[\Cov_{X_{\tau}}(X_{\tau} \mid X_{P})\right]\right)
\end{alignat*}
\end{small}
The third inequality is due to (\ref{eq:thmcond}). The last equality holds since $\PA(\tau) \subseteq P$, and hence, $X_\tau$ is independent of $X_{P\setminus \PA(\tau)}$ conditioned on $X_{\PA(\tau)}$. 
\end{proof}
\vspace{-3mm}
\end{proof}

The following corollary is immediate.

\begin{corollary}\label{cor:ampsem}
Suppose $X$ corresponds to an additive noise model generated by a chain graph $\caC$, i.e.:
\[
X_{\tau} = f_\tau(X_{\PA(\tau)}) + Z_\tau,
\]
where the noise term $Z_\tau$ is independent of $X_{\PA(\tau)}$, for all chain components $\tau$ of $\caD$.

Then, given the chain component decomposition, a topological ordering of $\caD$ is identifiable from $X$ if
there exists a topological ordering $\pi$ of $\caD$ such that
\[
\det(\Cov(Z_\tau)) \leq \det(\Cov(Z_{\tau'}))
\]
for all chain components $\tau \prec_\pi \tau'$. 
\end{corollary}
\ignore{
\noindent In our first step in \Cref{lm:1}, we will prove that source chain components in $\mathcal{G}(D,\Sigma)$ are characterized by minimal determinant. This can be done by recursively identifying source chain components for $\mathcal{G}(D)$ and subgraphs. 
Then we show that alternatively one may minimize determinant to identity a sink node.
This Lemma could be used to prove \Cref{th:3} and \Cref{th:4}.


\begin{lemma}
Let $X \sim (D_X, \Sigma)$ forms a chain graph with $A \in V$ an ancestral set in $\mathcal{G}$. 
For chain component $\tau$, 
if $\PA(\tau)=  \emptyset$, then $\Det (\mathbb{E} (\Cov(X_{\tau})) )= \Delta$; If $\PA(\tau) \neq \emptyset$, then $\Det( \mathbb{E}( \Cov(X_{\tau}))) > \Delta$. This means the source chain components in $\mathcal{G}(X)$ are characterized by minimal determinant. 
Besides,
If $\Det( \mathbb{E}[ \Cov (X_{\tau} |X_A) ]) \equiv \Delta$
does not depend on $\tau$. Then for any $\tau \notin A $, if $\PA(\tau) \in A$, then $\Det (\mathbb{E}( \Cov ( X_{\tau} |X_A) ))= \Delta$, otherwise, $\Det (\mathbb{E}( \Cov ( X_{\tau} |X_A) )) > \Delta$.
\label{lm:1}
\end{lemma} 

\noindent Given $X_{\tau} = B X_{Pa_{(\tau)}} + N_{\tau}, N \sim \mathcal{N}(0, \Sigma_{\tau})$, we assume $\Det(\Sigma_{\tau})
= \Delta $  and is fixed. 
The proof is given as follows using the method of mathematical induction: \\

\noindent \textbf{Step (1)} For any chain component $k \in V \setminus \{ \tau_1 \}$, we have: 
\begin{equation*}
\Det(\Cov(X_1)) \leqslant \Det(\Cov(X_{\tau})), \forall \tau  \neq 1\\
\end{equation*}

\begin{proof}
\begin{equation*}
\begin{split}
& \Det(\Cov(X_1)) = \Det(\Sigma_1) = \Delta \\
& \Det(\Cov(X_{\tau})) = \Det(\Sigma_{\tau} + \underbrace{\Cov(B X_{Pa_{(\tau)}}) }_{P.S.D.} )   \geqslant \Det(\Sigma_{\tau}) = \Delta = \Det(\Cov(X_1))
\end{split}
\end{equation*}
\end{proof}

\noindent \textbf{Step (m-1)} For the $(m-1)^{th}$ element of the ordering, assume that the first $m-1$ elements of the ordering and their parents are correctly recovered. \\
\textbf{Step (m)} Now consider the $m^{th}$ element of the ordering and its parents. Suppose $\tau_1, \dots, \tau_m$ are already identified. For $\tau_k \in \{ \tau_{m+1}, \dots, \tau_{p}\}, \forall k > m$, we have:
\begin{equation*}
    \Det(\mathbb{E}[\Cov(X_{\tau_m} | X_{\tau_{1: m-1}} )]) \leqslant \Det(\mathbb{E}[ \Cov(X_{\tau_k} | X_{\tau_{1: m-1}})])
\end{equation*}
\begin{proof}
\begin{equation*}
\begin{split}
   & \:\:\:\:\:\Cov(X_{\tau_k}) = \mathbb{E}[(X_{\tau_k}- \mu)(X_{\tau_k} - \mu)^T] = \mathbb{E}[X_{\tau_k}  X_{\tau_k}^T] - \mu \cdot \mu^T \\ 
   & \:\:\:\:\:\Det(\mathbb{E}[\Cov(X_{\tau_m} | X_{\tau_{1: m-1}} )]) = \Delta \\ 
   & \:\:\:\:\:\Det(\mathbb{E}[\Cov(X_{\tau_k} | X_{\tau_{1: m-1}} )])  = \Det(\mathbb{E}[ \mathbb{E}[ X_{\tau_k} X_{\tau_k}^T | X_{\tau_{1: m-1}} ]  - \mathbb{E}[ X_{\tau_k} | X_{\tau_{1: m-1}}] \cdot \mathbb{E}[ X_{\tau_k} | X_{\tau_{1: m-1}}] ^T    ]) \\ 
   & = \Det(\mathbb{E}[\underset{X_{pa(\tau_k)}}{\mathbb{E}} \mathbb{E}[X_{\tau_k} X_{\tau_k} ^T | X_{pa(\tau_k)} ] - \mathbb{E}[X_{\tau_k} | X_{\tau_{1: m-1}} ]  \cdot \mathbb{E}[X_{\tau_k} | X_{\tau_{1: m-1}} ] ^T] ) \\ 
    & = \Det(\mathbb{E}[ \underset{X_{pa(\tau_k)}}{\mathbb{E}}[\Cov(X_{\tau_k} | X_{pa(\tau_k)} )] + \underset{X_{pa(\tau_k)}}{\mathbb{E}}[\mathbb{E}[X_{\tau_k} | X_{pa(\tau_k)]}  ]\cdot \mathbb{E}[X_{\tau_k} | X_{pa(\tau_k)}  ] ^T] -\mathbb{E}[X_{\tau_k} | X_{\tau_{1: m-1}} ] \cdot \mathbb{E}[X_{\tau_k} | X_{\tau_{1: m-1}} ]^T ]) \\ 
    & = \Det(\mathbb{E}[ \underset{X_{pa(\tau_k)}}{\mathbb{E}}[\Sigma_{\tau_k}] +  \underset{X_{pa(\tau_k)}}{\mathbb{E}}[ \mathbb{E}[X_{\tau_k} | X_{pa(\tau_k)}  ] \cdot \mathbb{E}[X_{\tau_k} | X_{pa(\tau_k)}  ] ^T ] -{\mathbb{E}} [X_{\tau_k} | X_{\tau_{1: m-1}} ] \cdot \mathbb{E}[X_{\tau_k} | X_{\tau_{1: m-1}} ]^T ]) \\ 
    & = \Det(\Sigma_{\tau_k} + \underset{X_{\tau_{1: m-1}}}{\mathbb{E}}[ \underset{X_{pa(\tau_k)}}{\mathbb{E}} [ \mathbb{E} [X_{\tau_k} | X_{pa(\tau_k)}] \cdot \mathbb{E} [X_{\tau_k} | X_{pa(\tau_k)}]  ^T] \:|\: X_{\tau_{1: m-1}} ] -\underset{X_{\tau_{1: m-1}}}{\mathbb{E}} [\mathbb{E} [X_{\tau_K}]\cdot \mathbb{E}[X_{\tau_K} ]^T \:|\: X_{\tau_{1: m-1}} ] ) \\ 
    & =  \Det(\Sigma_{\tau_k} +  \underset{X_{\tau_{1: m-1}}}{\mathbb{E}}[ \underset{X_{pa(\tau_k)}}{\mathbb{E}} [ \mathbb{E} [X_{\tau_k} | X_{pa(\tau_k)}] \cdot \mathbb{E} [X_{\tau_k} | X_{pa(\tau_k)}]  ^T] \:|\: X_{\tau_{1: m-1}} ]  -  \\
    &\quad \quad \quad \underset{X_{\tau_{1: m-1}}}{\mathbb{E}} [ \underset{X_{pa(\tau_k)}}{\mathbb{E}} [ \mathbb{E}[X_{\tau_k} | X_{pa(\tau_k)}] ] \cdot \underset{X_{pa(\tau_k)}}{\mathbb{E}} [\mathbb{E}[X_{\tau_k} | X_{pa(\tau_k)}]  ]^T \:|\: X_{\tau_{1: m-1}} ] ) \\
    & = \Det(\Sigma_{\tau_k} +   \underbrace{\underset{X_{\tau_{1: m-1}}}{\mathbb{E}}[ \Cov( \mathbb{E}[X_{\tau_k} | X_{pa(\tau_k)}  ]  ) ]}_{P.S.D.} ) \\ 
    & \geqslant \Det(\Sigma_{\tau_k}) 
    = \Det(\mathbb{E}[\Cov(X_{\tau_m} | X_{\tau_{1: m-1}} )]  )
\end{split}
\end{equation*}
 This proves that $\Det( \mathbb{E}[\Cov (X_{\tau_m} | X_{\tau_{1:m-1}}) ] ) \leqslant \Det( \mathbb{E} [ \Cov (X_{\tau_k} | X_{\tau_{1:m-1}})])$, which completes the proof. 
\end{proof}

Step (1) shows that we can identify the source node due to the minimum determinant. In step (m), we show that by conditioning on an ancestral set, one can recover a chain graph with equal determinant assumption whose graph has the entire ancestral set removed in a iterative way. Specifically, in \Cref{th:3} and \Cref{th:4}, the covariance matrix shows the information of how variables in each chain components vary with each other, and the determinant of the covariance matrix gives the measure of magnitude of how much the variables vary with each other. In short, equal determinant can help identfy the topology sort of chain graph $\mathcal{G}$.

Furthermore, the same proof works for any generalized matrix function, including not only determinant but also permanent, trace, and product of diagonal elements. 
As mentioned in Equation (1.1) of \cite{paksoy2014inequalities}: the inequality $g(A+B) \geqslant g(A) + g(B)$ holds for any "generalized matrix function" that includes the determinant, permanent, and product of diagonal elements. For positive definite matrices $A$ and $B$, we get $g(A+B) > g(A)$, because $g(B) > 0.$ (This follows by applying Theorem 3.3 of \cite{marcus1965generalized} with m=n.)
}

\section{General Identifiability}\label{sec:general}
In this section, we establish identifiability conditions for recovering both the chain components as well as the DAG structure of chain graphs from the generated probability distribution. Here, by identifiability, we mean that the partitioning into chain components and the topological order on the chain components are uniquely specified. The exact set of directed and undirected edges can then be recovered using standard variable selection methods (as described in Appendix A of \cite{GDA20}). 

\vspace{-2mm}
\begin{algorithm}
$ P \leftarrow \emptyset$\;
$i \leftarrow 1$\;
$\tau_1 = \arg\min_{S \subseteq V,  S \neq \emptyset}\det (\Cov(X_S))$ \;
{$P \leftarrow P \cup \tau_1$}\;
\While{$V \setminus P \neq \emptyset$}{
$\tau_i \leftarrow \arg\min_{S\subseteq V \setminus P, S \neq \emptyset} \det(\Cov(X_{S} \mid X_P))$\;
$P \leftarrow P \cup \tau_i$\;
$i \leftarrow i+1$\;
}
Return the topological sort $(\tau_1, \dots, \tau_i)$
\caption{Infinite sample algorithm for learning the topological order of a chain graph with unknown chain components.}
\label{algo:pop2}
\end{algorithm}
\vspace{-3mm}

\ignore{

\begin{theorem}
Suppose the random variable $X$ is generated by an AMP-CG $G$ with unknown structure. Then, $G$ is identifiable from $X$ if the following two conditions hold:
\begin{enumerate}
\item[(i)]
There exists a topological order $\pi$ on the chain components such that:
\ignore{ 
\[
\det\left(\Exp_{X_{\PA(\tau)}}\Cov_{X_\tau} (X_{\tau} \mid X_{\PA(\tau)})\right) \leq \det\left(\Exp_{X_{\PA(\tau')}}\Cov_{X_{\tau'}}(X_{\tau'} \mid X_{\PA(\tau')})\right)
\]
}
\begin{small}
\begin{align*}
 \hspace{-5mm}  \det\left(\Exp_{X_{\PA(\tau)}}\Cov_{X_\tau} (X_{\tau} \mid X_{\PA(\tau)})\right) \leq \det \left(\Exp_{X_{\PA(\tau')}}\Cov_{X_{\tau'}}(X_{\tau'} \mid X_{\PA(\tau')})\right)
\end{align*}
\end{small}
for all chain components $\tau \prec_\pi \tau'$.
\item[(ii)] For all chain components $\tau$:
\[
\lambda_{\max}\left( \Exp_{X_{\PA(\tau)}}\Cov_{X_\tau} (X_{\tau} \mid X_{\PA(\tau)}) \right) <1.
\]
\end{enumerate}

\end{theorem}
}
\input{major/general.tex}

%% file: major/general.tex

\begin{theorem}\label{TM:4.1}
Suppose the random variable $X$ is generated by an AMP-CG $\mathcal{C}$ with unknown structure.
Then, $\mathcal{C}$ is identifiable from $X$ if the following three conditions hold:
\begin{enumerate}
\item[(i)] For all chain components $\tau$ and all non-empty proper subsets $S \subset \tau$:
\[\det ( \Cov (X_s \mid X_{\tau \setminus s}, X_{\PA(\tau)})) < 1.\]
\item[(ii)]
For all chain components $\tau$:
\[\det (\Cov (X_{\tau} \mid X_{\PA(\tau)}))  > 1.\]
\item[(iii)]
There is a topological order $\pi$ on the chain components such that for all $\tau \preceq_\pi \tau'$.
:
\[
\det (\Cov (X_{\tau} \mid X_{\PA(\tau)})) \leq \det (\Cov (X_{\tau'} \mid X_{\PA(\tau')})).
\]
\end{enumerate}
\end{theorem}

\begin{proof}
For simplicity, suppose that there is a unique topological order $\tau_1 \preceq \tau_2 \preceq \cdots \preceq \tau_m$ for the components outside $P$. The proof easily extends to the general case.
Let $\tau_{<i} = \tau_1\cup \cdots \cup \tau_{i-1}$. 

Consider any {\em non-empty} subset $S$ that is disjoint from $P$. We claim:
\begin{align}
&\det(\Cov (X_S \mid X_P) )\nonumber \\ & = \prod_{i=1}^m \det (\Cov ( X_{S \cap \tau_i}\mid X_{S \cap \tau_{<i}}, X_P))\label{eq:claim_1} \\
& \ge \prod_{i=1}^m \det(\Cov( X_{S \cap \tau_i} \mid X_{\PA(\tau_i)}))\label{eq:claim_2} \\
& \ge \prod_{i: S \cap \tau_i \neq \emptyset} \det(\Cov(X_{\tau_i} \mid X_{\PA(\tau_i)}) )\label{eq:claim_3}\\
& \geq \det(\Cov(X_{\tau_1} \mid X_{\PA(\tau_1)})) \label{eq:claim_4}
\end{align}
(\ref{eq:claim_1}) is a consequence of Fact \ref{FACT:FACTORIZATION}. 
To prove (\ref{eq:claim_2}), we invoke the law of conditional covariance (Fact \ref{fact:tcov}):
\begin{align*}
&\det(\Cov (X_{S\cap \tau_i} \mid X_{S \cap \tau_{<i}}, X_P ))  \\
&= \det( \Cov(X_{S \cap \tau_i} \mid X_{\PA(\tau_i)}, X_{S \cap \tau_{<i}}, X_P) 
+ \Cov(\mathbb{E}[ X_{S \cap \tau_i} \mid X_{\PA(\tau_i)}] \mid X_{S \cap \tau_{<i}}, X_P ) )  \\
&\ge \det(\Cov(X_{s\cap \tau_i} \mid X_{\PA(\tau_i)}) )
\end{align*}
The last inequality uses the positive semi-definiteness of covariance matrices and super-additivity of the determinant.
The proof of (\ref{eq:claim_3}) uses Fact \ref{FACT:FACTORIZATION} as follows:
\begin{align*}
& \det(\Cov(X_{\tau_i} \mid X_{\PA(\tau_i)}) ) \\
&= \det(\Cov(X_{S\cap \tau_i} \mid X_{\PA(\tau_i)}))\cdot \det(\Cov(X_{\tau_i \setminus S} \mid X_{S \cap \tau_i}, X_{\PA(\tau_i)})) \\
&\le  \det(\Cov(X_{S\cap \tau_i} \mid X_{\PA(\tau_i)}) )
\end{align*}
using condition (iii) of Theorem \ref{TM:4.1}. (The last inequality is non-strict because $\tau_i \setminus S$ may be empty.) The inequality (\ref{eq:claim_4}) follows from conditions (i) and (ii) of Theorem \ref{TM:4.1}.

For (\ref{eq:claim_4}) to be an equality, $S$ must be contained in exactly one component $\tau$.  For (\ref{eq:claim_3}) to be an equality, $S$ must equal $\tau$. For (\ref{eq:claim_2}) to be an equality, the parents of $S\cap \tau = S$ must be contained in $P$, and hence $S=\tau = \tau_1$. 
\end{proof}
\ignore{
\section{Additional Motivation for Conditions of Theorem \ref{TM:4.1}}

There is a geometric way to view the conditions in Theorem \ref{TM:4.1}, which substantiates the intuition that they require each chain component to cluster together while having large variance as a whole.

Recall that for any matrix $M$, $\det(M)$ corresponds to the volume of the parallelepiped spanned by the rows of $M$. Let the chain components be denoted $\tau_1, \dots, \tau_k$ in a topological order. For $i = 1, \dots, k$, let $M_i$ denote the covariance matrix of $X_{\tau_i} \mid X_{\tau_1 \cup \cdots \cup \tau_{i-1}}$, and let $M$ denote the full covariance matrix, $\Cov(X_{\tau_1 \cup \cdots \cup \tau_k})$. From Fact   \ref{FACT:FACTORIZATION},
\begin{equation}\label{eq:orthogonal}
\det(M) = \det(M_1) \cdots \det(M_k).
\end{equation}

Let $V_i$ denote the set of row vectors of $M_i$, and we identify $V_i$ with the parallelepiped it spans. Due to Equation \ref{eq:orthogonal}, we can view each $V_i$ as residing in a subspace orthogonal to the spans of other $V_j$'s, so that their volumes just multiply with each other. (Alternatively, construct a block diagonal matrix $M'$ where the $i$'th block on the diagonal is $M_i$; clearly, $\det(M) = \det(M')$.)

In this language,  Condition (ii) in Theorem \ref{TM:4.1} says that the volume of each $V_i$ is more than $1$, and condition (iii) says that the volumes are non-decreasing with $i$. Condition (i) says that for any $V_i$, the volume of any sub-parallelepiped is larger than the volume of the  whole. Intuitively, this means that the vectors in $V_i$ form very small angles with each other, so that the volumes keep decreasing as more vectors are added.

}

Informally speaking, for any subset $S$, given its complementary set and parents union of $\tau$ in $\mathcal{C}$, we require the variables in each chain component to be tightly correlated.  Besides, given the union of the parents of chain components $\tau$, 
we require the clustered variables in each chain component to have large generalized variance. The third condition is the same one imposed in Section \ref{sec:id}.

There is a geometric way to view the conditions in Theorem \ref{TM:4.1}, which substantiates the intuition that they require each chain component to cluster together while having large variance as a whole. Recall that for any matrix $M$, $\det(M)$ corresponds to the volume of the parallelepiped spanned by the rows of $M$. Let the chain components be denoted $\tau_1, \dots, \tau_k$ in a topological order. For $i = 1, \dots, k$, let $M_i$ denote the covariance matrix of $X_{\tau_i} \mid X_{\tau_1 \cup \cdots \cup \tau_{i-1}}$, and let $M$ denote the full covariance matrix, $\Cov(X_{\tau_1 \cup \cdots \cup \tau_k})$. From Fact   \ref{FACT:FACTORIZATION},
\begin{equation}\label{eq:orthogonal}
\det(M) = \det(M_1) \cdots \det(M_k).
\end{equation}

Let $V_i$ denote the set of row vectors of $M_i$, and we identify $V_i$ with the parallelepiped it spans. Due to Equation \ref{eq:orthogonal}, we can view each $V_i$ as residing in a subspace orthogonal to the spans of other $V_j$'s, so that their volumes just multiply with each other. (Alternatively, construct a block diagonal matrix $M'$ where the $i$'th block on the diagonal is $M_i$; clearly, $\det(M) = \det(M')$.) In this language,  Condition (ii) in Theorem \ref{TM:4.1} says that the volume of each $V_i$ is more than $1$, and condition (iii) says that the volumes are non-decreasing with $i$. Condition (i) says that for any $V_i$, the volume of any sub-parallelepiped is larger than the volume of the  whole. Intuitively, this means that the vectors in $V_i$ form very small angles with each other, so that the volumes keep decreasing as more vectors are added.

\paragraph{Computational Efficiency.}
{It is known that} Algorithm \ref{algo:pop2} can be implemented in polynomial time. This is because the optimization problems in lines 3 and 5 of the pseudocode correspond to submodular function minimization, as explained in Section \ref{sec:tover}.  Solving submodular function minimization is in polynomial time (see, e.g., \cite{iwata2008submodular}).

%% file: major/expt.tex

\section{Experiments}

\begin{figure}
\centering
\begin{tabular}{|c|c|}
\hline
$\mathcal{C}$&$\mathcal{C}'$\\
\hline
\begin{tikzpicture}[inner sep=1mm]
\node at (0,0.5) (A) {$A$};
\node at (1,0.5) (B) {$B$};
\node at (0,-1) (C) {$C$};
\node at (1,-1) (D) {$D$};
\path[->] (A) edge (B);
\path[->] (A) edge (C);
\path[->] (A) edge (D);
\path[->] (B) edge (D);
\path[-]  (C) edge (D);
\end{tikzpicture}
&
\begin{tikzpicture}[inner sep=1mm]
\node at (0,0) (A) {$A$};
\node at (1,0) (B) {$B$};
\node at (0,-1) (C) {$C$};
\node at (1,-1) (D) {$D$};
\node at (0.5,-1.5) (ECD) {$\epsilon_{CD}$};
\node at (-1,0) (EA) {$\epsilon_A$};
\node at (2,0) (EB) {$\epsilon_B$};
\node at (-1,-1) (EC) {$\epsilon_C$};
\node at (2,-1) (ED) {$\epsilon_D$};
\path[->] (EA) edge (A);
\path[->] (EB) edge (B);
\path[->] (EC) edge (C);
\path[->] (ED) edge (D);
\path[->] (A) edge (B);
\path[->] (A) edge (C);
\path[->] (A) edge (D);
\path[->] (B) edge (D);
\path[-] (EC) edge [bend right] (ECD);
\path[-] (ED) edge [bend left] (ECD);
\end{tikzpicture}\\
\hline
\end{tabular}\caption{Synthetic data generation. Undirected edges correspond to correlated noise.}\label{fig:example3}
\vspace{-0.5cm}
\end{figure}
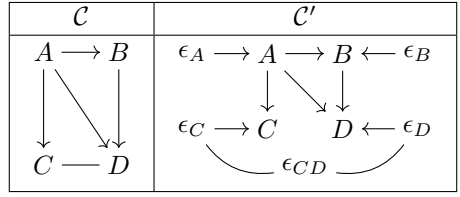

In this section, we compare the performance of Algorithm \ref{algo:pop1} and Algorithm \ref{algo:pop2} on synthetic datasets to state-of-the-art methods for AMP-chain graph structure learning. Recall that as we showed in Theorem \ref{THM:ID}, the DAG on the chain components of an AMP chain graph is identifiable if (\ref{eq:thmcond}) is satisfied for a positive and super-additive family $d_\tau$. Here, we let $d_\tau$ be the determinant operator, and hence dub our algorithm as Determinant of Covariance (DCOV).

\textbf{Synthetic Data Generation:} 
To generate the chain graph $\mathcal{C}$, in our first step, an undirected graph $\mathcal{G}$ with $n$ nodes is generated by using the Erd\H{o}s R\'enyi (ER) model with an expected neighbor size $s=2$ and then symmetrizing. Given the number of chain components $c$, we split the interval $[ 1, n ] $ into $c$ equal-length sub-intervals $[I_1, \dots, I_c]$ so that variable sets for each sub-interval forms chain components $\tau_1, \dots, \tau_c$. Meanwhile, for any $(i, j)$ pair, we set $\mathcal{C}_{i,j} = 0$ if $\exists i \in I_\ell, j \in I_m, \ell > m$.
Given the binary adjacency matrix $\mathcal{C}$, we generate the matrix $M$ of edge weights by $M_{i,j} \sim U(-1.5,-0.5] \cup U[0.5, 1.5)$ if $\mathcal{C}_{i,j}\neq 0$  and $M_{i,j}=0$ otherwise. 

The observational i.i.d. data $X_\tau = M_\tau X_{\PA(\tau)} + Z_\tau$ is generated with a sample size $n=1000$ and a variable size $d \in \{10, 20, 30, 40, 50\}$. $Z_\tau$ is an independent multivariate Gaussian drawn from $N(0, \Sigma_\tau)$ where $\Sigma_\tau$ is generated randomly with $\det(\Sigma_\tau)=1$, satisfying the assumption of Corollary \ref{cor:ampsem}. \Cref{fig:example3} illustrates how the synthetic AMP chain graph data is generated.

\textbf{Baseline Algorithms:}
We compare our \textsc{DCOV} method against 
the \textsc{PC-Like} (\cite{pena2012learning, pena2014learning, pena2016learning}), \textsc{LCD-Like} (Learn Chain Graphs via Decomposition), and \textsc{LDCG} algorithm (learn the largest deflagged graph based on the results of \textsc{LCD-Like} algorithm) ( \cite{javidian2020amp}). We use default
parameters among those baseline algorithms in order to avoid skewing the results in favour of any particular algorithm as a result of hyperparameter tuning\footnote{The implementation of baseline algorithms is available at 
\url{https://github.com/majavid/AMPCGs2019}.}.
All the baseline algorithms above are implemented using 
R-packages (licensed under GPL-2 or GPL-3) such as \textbf{ggm} (\cite{marchetti2006independencies}), \textbf{pcalg} (\cite{kalisch2012}), \textbf{mgcv} (\cite{wood2015package}), \textbf{np} (\cite{racine2020package}), and \textbf{lcd}(~\cite{ma2009package}). 
We use \textbf{rpy2} \citep{gautier2012rpy2} to access R-packages from Python and ensure that all algorithms can be compared in the same environment.
The experiments were conducted on an Intel Core i7-9750H 2.60GHz CPU.

\textbf{Implementation of \textsc{DCOV}:}
We implement Algorithm \ref{algo:pop2} in polynomial time using the Matlab toolbox \textbf{“Submodular Function Optimization”} \citep{krause2010sfo}. We use MATLAB Engine API for Python   to access Matlab-packages from Python. 
Each iteration of Algorithm \ref{algo:pop1} and Algorithm \ref{algo:pop2} needs to estimate the conditional covariance of the remaining chain components given those found so far. Our estimator of the conditional covariance is very similar to that considered by \cite{GDA20} for Bayes networks. In particular, like them, we run a \emph{gam} regression to estimate conditional expectations. We set the p-value with significance level of 0.001 for determining the parents of the node.

\textbf{Quantitative Experiment Results:}
In our experiment, 
we use Structural Hamming Distance (SHD) as the evaluation metric. Figure \ref{fig:exp} reports SHD of our proposed DCOV and other algorithms. The results are averaged over 20 independent repetitions.
As shown in Figure \ref{fig:exp},  DCOV, under known chain component conditions, shows superior performance compared with all other baselines by wide margins. 
Under unknown chain component conditions, DCOV outperforms \textsc{LDCG} and \textsc{LCD}, and is comparable to the \textsc{PC-Like} algorithm.  
One limitation of this work is the lack of real datasets that can be modeled by chain graphs.

\begin{figure*}
\centering    
\subfigure[Algorithm 1]{\label{fig:cg_a}\includegraphics[width=65mm]{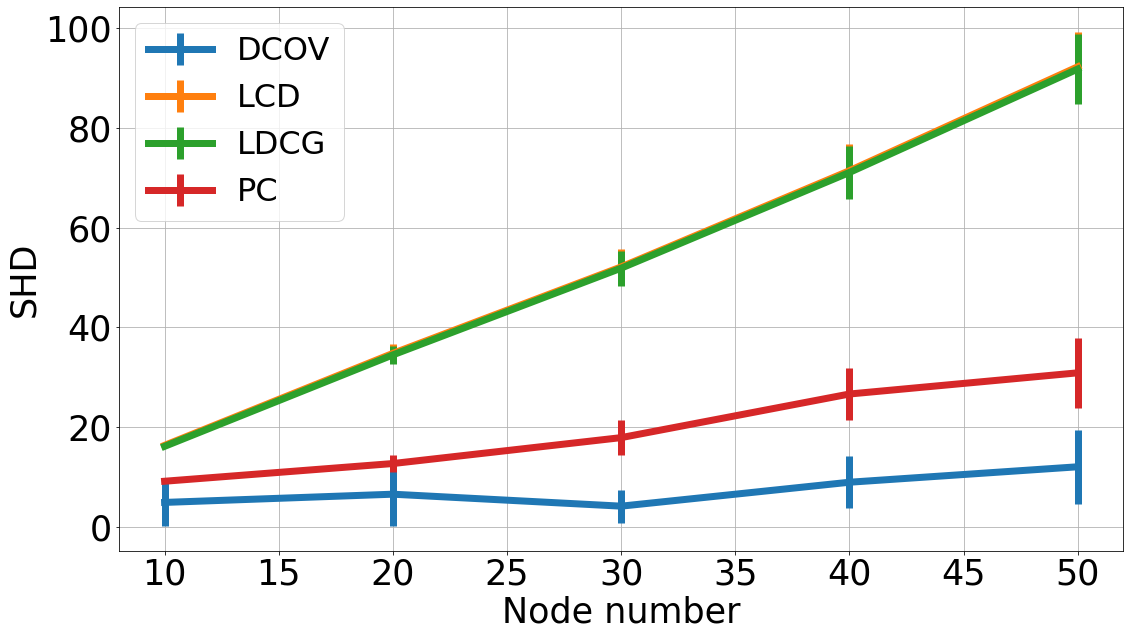}}
\hspace{2mm}
\subfigure[Algorithm 2]{\label{fig:cg_c}\includegraphics[width=65mm]{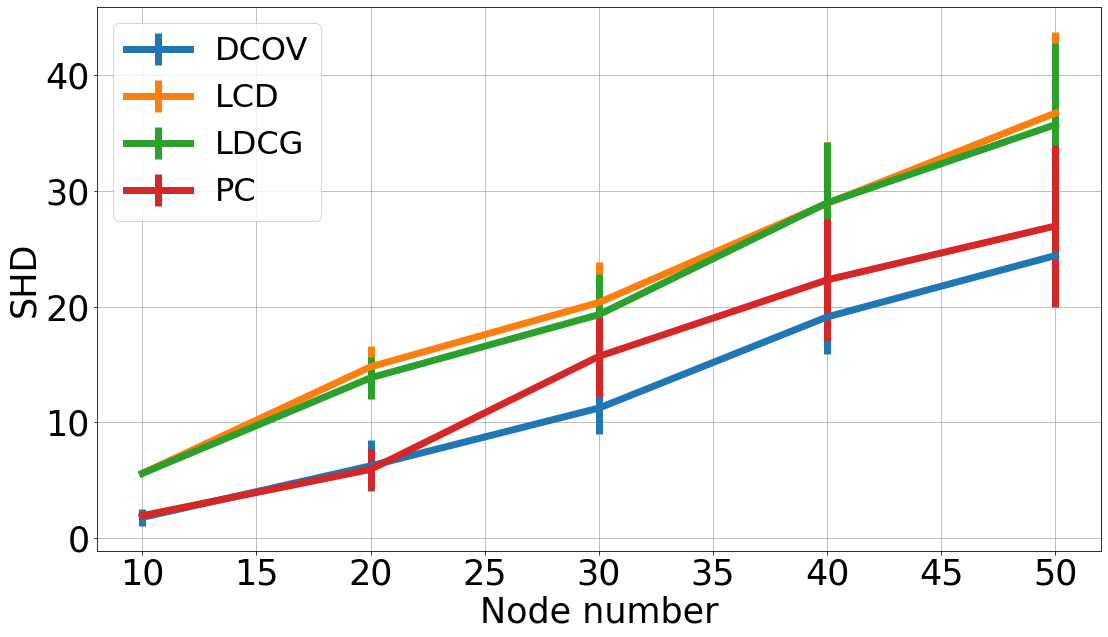}}
\caption{SHD performance (lower is better)}
\vskip -0.5cm
\label{fig:exp}
\end{figure*}

%% file: major/conc.tex
\section{Conclusion}
In this work, we address the problem of recovering AMP chain graph in polynomial time from observational data,
and we proposed two algorithms for both known and unknown chain components to handle the problem. In our experiments, we implement the DCOV algorithm over known chain components. As future work, we are also interested in exploring a score-based approach for chain graph structure learning from observational data. 